\documentclass[11pt]{article}
\usepackage{hyperref}
\usepackage{tikz}

\usepackage{tikz-cd}

\usepackage{amsmath,amsthm,amsfonts,amssymb,amsopn,amscd,bm} 
\usepackage{color}

\def\b{\beta}

\def\e{\varepsilon}
\def\g{\gamma}

\def\l{\lambda}

\def\setminus{\smallsetminus}

\def\A{{\cal A}}
\def\B{{\cal B}}

\def\C{{\cal C}}
\def\D{{\cal D}}

\def\M{{\cal M}}

\def\R{{\cal R}}

\def\H{{\cal H}}

\def\S{{\cal S}}

\def\f{{\varphi}}
\def\s{{\sigma}}

\def\l{{\lambda}}
\def\x{{{h}}}

\def\PSL{{{\rm PSL}(2,\mathbb R)}}

\def\S2{S^{1(2)}}

\def\RR{\mathbb R}

\newtheorem{theorem}{Theorem}[section]
\newtheorem{lemma}[theorem]{Lemma}

\newtheorem{corollary}[theorem]{Corollary}

\newtheorem{proposition}[theorem]{Proposition}

\newtheorem{identity}[theorem]{Identity}

\theoremstyle{remark} \newtheorem{remark}[theorem]{Remark}

\newcommand{\ben}{\begin{equation}}
\newcommand{\een}{\end{equation}}

\def\setminus{\smallsetminus}

\def\PSL{PSU(1,1)}

\def\Z{{\mathbb Z}}

\def\SL2{{{\rm SL}(2,\R)}}

\def\PSL2{{{\rm PSL}(2,\Reali)}}

\def\U1{{{\rm V}(1)}}
\def\SU2{{{\rm SV}(2)}}

\def\SU{{{\rm SU}}}

\def\A{{\mathcal A}}
\def\B{{\mathcal B}}
\def\C{{\mathcal C}}
\def\D{{\mathcal D}}

\def\H{{\mathcal H}}

\def\M{{\mathcal M}}

\def\S{{\mathcal S}}
\def\T{{\mathcal T}}

\def\bx{{\bf x}}

\def\bp{{\bf p}}

\def\cD{{\cal D}}

\def\cS{{\cal S}}

\def\bC{{\mathbb C}}

\def\bR{{\mathbb R}}
\def\RR{{\mathbb R}}

\def\b{\beta}
\def\g{\gamma}

\def\l{\lambda}

\def\x{\xi}

\def\s{\sigma}

\def\f{\varphi}

\def\ov{\overline}

\title{\Huge{The massless modular Hamiltonian
}}

\author{{\sc Roberto Longo\thanks{Supported by the ERC Advanced Grant 669240 QUEST ``Quantum Algebraic Structures and Models'', MIUR FARE R16X5RB55W  QUEST-NET and GNAMPA-INdAM. \eject
E-mail: {longo@mat.uniroma2.it, morsella@mat.uniroma2.it}
}
}, {\sc Gerardo Morsella}\\
Dipartimento di Matematica,
Universit\`a di Roma Tor Vergata,\\
Via della Ricerca Scientifica, 1, I-00133 Roma, Italy
}

\date{}
\begin{document}

\maketitle

\begin{abstract}
We compute the vacuum local modular Hamiltonian associated with a space ball region in the free scalar massless Quantum Field Theory. We give an explicit expression on the one particle Hilbert space in terms of the higher dimensional Legendre differential operator. The quadratic form of the massless modular Hamiltonian is expressed in terms of an integral of the energy density with the parabolic distribution. We then get the formula for the local entropy of a wave packet. This gives the vacuum relative entropy of a coherent state on the double cone von Neumann algebras associated with the free scalar QFT. 
Among other points, we provide the passivity characterisation of the modular Hamiltonian within the standard subspace setup. 
\end{abstract}

\newpage

\section{Introduction}
In this paper, we provide the formula for the vacuum modular Hamiltonian associated with the free massless scalar quantum field in a bounded spacetime region (double cone). 

\smallskip
\noindent
{\it Background for the modular Hamiltonian.} Let $\M$ be a von Neumann algebra and $\f$ a faithful normal state on $\M$. As is well known, the Tomita-Takesaki modular theory provides us with a canonical one-parameter automorphism group of $\M$ associated with $\f$, the modular group $\s^\f$. Thus, the quantum system $\M$ is equipped with an intrinsic evolution $\s^\f$, that is characterised by the KMS thermal equilibrium condition \cite{Tak}.
In the GNS representation, the modular group is implemented by a unitary one-parameter group $\Delta_\f^{is}$, whose generator $\log\Delta_\f$ is the modular Hamiltonian associated with $\f$ (see \cite{Lexp} for a discussion on the matter). 

In Quantum Field Theory, for each spacetime region $O$, we have the von Neumann algebra $\A(O)$ of observables localised in $O$. By the Reeh-Schlieder theorem, the vacuum vector is cyclic and separating for $\A(O)$ if both $O$ and its causal complement have non empty interiors, see \cite{Haag}. Thus the restriction of vacuum state $\f$ to $\A(O)$ is faithful; the modular Hamiltonian associated with $\A(O)$ and $\f$ gives us a local Hamiltonian $\log \Delta_O$. 

The problem of computing  $\log \Delta_O$ is then natural. Among other motivations, the modular Hamiltonian is related to Araki's relative entropy \cite{Ar76}, that is recently playing an important role in Quantum Field Theory, also in relation with quantum entropy/energy inequalities; an account of this wide research topic goes beyond the purpose of this introduction (see \cite{ABCH, CF, L18, LX18, LXbek, Wit} and refs. therein).

If $W$ is a wedge region, there is an important model independent result  \cite{BW}: $\Delta_W^{-is}$ is identified with the $2\pi$-rescaled boost unitary transformations leaving $W$ globally invariant, a result rich of consequences, see \cite{Haag}. 
For a different spacetime region $O$, the understanding of the local modular structure is definitely more problematic because, in general, $\Delta_O^{is}$ has no geometric description due to the lack of enough spacetime symmetries. A basic issue concerns the case where $O$ is a double cone, the causal envelop of a time-zero ball.
In the free massless QFT, the geometric description of the double cone vacuum modular group was derived in \cite{HL}. The general conformal case was later analysed solely in terms of the local von Neumann algebras \cite{FGa, BGL}. 
Among other results, we mention here the forward light cone case in the free massless QFT \cite{B77} and the approximate local estimates in \cite{F85}. 

The modular theory has a version for standard subspaces, see \cite{LRT, RvD, LN}. Let $\H$ be a complex Hilbert space and $H\subset \H$ a standard subspace, i.e. $H$ is a real linear, closed subspace of $\H$ such that $H\cap iH' = \{0\}$ and ${\ov{H + iH}} = \H$, with $H'$ the symplectic complement of $H$. The modular operator $\Delta_H$ associated with $H$ is a canonical positive, non-singular selfadjoint operator on $\H$ associated with $H$ that satisfies 
\[
\Delta^{is}_H H = H\, , \quad s\in \mathbb R\, .
\]
In a free scalar QFT, the vacuum modular (unitary) group associated with the von Neumann algebra of a region $O$ is the second quantisation of a unitary one-parameter group in the one-particle Hilbert space, indeed of the modular group of the local standard subspace $H(O)$. We shall henceforth denote by $\Delta_O$ the modular operator $\Delta_{H(O)}$ on the one particle Hilbert space. Also, we set $H(B) = H(O)$, $\Delta_B= \Delta_O$ if $O$ is  the double cone with basis the unit space ball centred at origin $B$. Hence, in this paper
\[
\text{Modular Hamiltonian of $B$} = \log\Delta_B\, .
\]
In the massless case,  $\Delta_B^{is}$ is associated with a one parameter group of conformal transformations that globally preserve $O$ \cite{HL}. Nonetheless, the description of $\log \Delta_B$ has not been worked out so far. We shall assume that the spatial dimention $d$ is greater than one. The case $d=1$ is essentially the same on a natural subspace of test functions, see \cite[Thm. 5.9]{L21}. 

In terms of the wave Cauchy data, we shall see that the local massless modular Hamiltonian is given by
\ben\label{Id}
\log\Delta_B = -2\pi i\left[\begin{matrix} 0 & \frac12(1-r^2) \\
\frac12 (1-r^2)\nabla^2 - r \partial_r - D&0\end{matrix}\right]\, ,
\een
with $D = (d-1)/2$ the scaling dimension of the free scalar field. Namely
\ben
\log\Delta_B = -2\pi i\left[\begin{matrix} 0 & M \\
L&0\end{matrix}\right]\, ,
\een
with
\begin{align}\label{LM}
&M =\ \text{Multiplication operator by}\ \,\frac12(1-r^2) \, ,\\
&L =\ \text{Legendre operator}\ \,\frac12 (1-r^2)\nabla^2 - r \partial_r - D\, .
\end{align}
As we shall see,
the right-hand side of \eqref{Id} gives indeed an essentially selfadjoint operator on the one particle Hilbert space on the smooth, compactly supported function domain. 

We shall derive the following formula for the massless modular Hamiltonian of $B$ in terms of the classical stress-energy tensor $T$:
\ben\label{lH}
-(\Phi, \log \Delta_{B} \Phi) 
= 2\pi\int_{x_0 = 0} \frac{1 - r^2}{2} \langle T_{00} \rangle_{\Phi}(x) dx  + \pi D \int_{x_0 = 0} \Phi^2 dx \, ,
\een
with $\Phi$ a real wave with smooth, compactly supported Cauchy data. Here $T_{00}$ is energy density
\ben\label{ed}
T_{00} = \frac12 \big( (\partial_0 \Phi)^2  + |\nabla_{\bx} \Phi|^2    \big)\ .
\een

The right-hand side of \eqref{lH} is similar to a formula for the modular Hamiltonian sketched by Casini, Huerta and Myers \cite[(2.23)]{CHM} in terms of the QFT stress-energy tensor, when the Cauchy data are supported in $B$. 

\smallskip
\noindent
{\it The measure of information.}
One main consequence of our analysis is a formula for the entropy density carried by a wave packet. Let's explain the framework. 

Suppose $\Phi$ is a scalar wave packet, a solution of the wave equation $\square\Phi =0$.  At a given time, we can measure the signal contained say in space ball $B$. The quantity $S_\Phi$ that represents the mean information stored by $\Phi$ in $B$ at that time is called the entropy of $\Phi$ with respect to $B$ and has been introduced in \cite{L18, L19, CLR19}, although different space regions (wedges) were there considered. 

With $\H$ a complex Hilbert space and $H\subset \H$ a factorial standard subspace, one defines the entropy of a vector $k\in \H$ with respect to $H$ as
\ben\label{Sintro}
S_k = \Im(k, P_H i \log\Delta_H\, k) \ .
\een
Here, $P_H : H + H' \to H$, $P_H: h + h' \mapsto h$,  is the cutting projection associated with $H$, $\Delta_H$ is the modular operator associated with $H$ and $H'$ is the symplectic complement of $H$.  

Motivated by Quantum Field Theory, one equips the waves' real linear space with a complex Hilbert space structure, where the imaginary part of the scalar product is given by the time-independent symplectic form
\[
\Im(\Phi,\Psi) = \frac12\int_{x_0 = t} \big( \Psi \partial_0 \Phi -  \Phi  \partial_0\Psi\big)d\bx \, .
\]
Then one considers the local net of standard subspaces associated with the resulting $\H$: $H(B)$ is the closure of the real linear space of waves with Cauchy data supported in $B$.  So one defines $S_\Phi$ as the entropy of the vector $\Phi$ with respect to $H(B)$. We have
\ben\label{Ss}
S_\Phi = \pi\int_{B} \big(\Psi\Phi'    - \Phi \Psi' \big)d\bx \  ,
\een
(time zero integral) where $\Psi = i\log\Delta_{B} \Phi$, the prime denotes the time derivative and $\Delta_{H(B)} = \Delta_B$. 

One then has to compute \eqref{Ss}. 
In \cite{L19, CLR19}, this computation has been worked out in the case of a wave for a half-space region (whose causal envelop is a wedge), a case motivated by the study of the Quantum Null Energy Condition inequality. 

From the point of view of information theory, it is however natural to consider the case that $B$ is a bounded region. Note that the computation in \cite{L19, CLR19} relies on the explicit knowledge of the modular Hamiltonian $\log \Delta_{W}$ \cite{BW}.

Now, $P_{H(B)}$ acts by cutting the Cauchy data \cite{CLR19}, so we have at our disposal all the ingredients to compute the local entropy $S_\Phi(R)$ of the wave packet $\Phi$ in the (causal envelop of) the radius $R$ space ball $B_R(\bar\bx)$ around the space point $\bar\bx$, at time $t$. We shall see that $S_\Phi(R)$
is the sum of two terms:
\begin{align*}
S_\Phi(R) &= \pi\int_{B_R(\bar\bx )} \frac{R^2 - r^2}{R} \langle T_{00}(t,\bx)\rangle_{\Phi} d\bx \qquad   \qquad\text{stress-energy tensor term}\\
& + \pi\frac{d-1}{2R}\int_{B_R(\bar\bx)} \Phi^2(t,\bx) d\bx    \quad\qquad\qquad \qquad \qquad \text{Born type term}
\end{align*}
with $r = |\bx - \bar\bx|$ and $T_{00}$ the energy density of $\Phi$ \eqref{ed}. 

We note here the appearance of the parabolic distribution $\frac12(1 - r^2)$ in the stress-energy term of
the formula for the modular Hamiltonian. We wonder about possible deep roots for this, somehow similarly to the appearance of the Wigner semicircular distribution \cite{Wig} in the Free Probability framework \cite{Voi}. The parabolic distribution in three-dimensional space
is a higher-dimensional generalisation of the Wigner semicircular distribution and is related to the marginal distribution of a spherical distribution. 

\smallskip
\noindent
{\it Content of this paper.} 
Our paper is organised as follows. First, we collect functional analytic results. Among these, we have 
a standard subspace version of the Pusz and Woronowicz \cite{PW} complete passivity characterisation, up to a proportionality constant, of the modular Hamiltonian. This reflects the second principle of thermodynamics. By our results, our local Hamiltonian is completely passive in our standard subspace sense. 
Then we describe the massless modular group in the wave setup, so we get an explicit description of the massless Hamiltonian. We then provide our local entropy formula for a wave packet and discuss some of the implications in Quantum Field Theory. 
Our general references for Operator Algebras and Quantum Field Theory are \cite{EK,Haag, Tak}. 

\section{Abstract preliminaries}

%
\subsection{Real linear invariant subspaces}
We now characterise the closed real linear subspaces of a complex Hilbert space $\H$ that are invariant for a one-parameter unitary group. 
The following proposition is indeed more general than what we shall need. 
\begin{proposition}\label{inv}
Let $\H$ be a Hilbert space, $H\subset \H$ a closed, real linear subspace and $A :D(A)\subset \H\to \H$ a selfadjoint operator. With $V(s) = e^{isA}$, $s \in \bR$, $K = iA$ and $\C$ (resp. $\B$) the real algebra of  complex, continuous functions $g$ on $\mathbb R$ vanishing at $\infty$ (resp. complex, bounded Borel functions on $\RR$) such that  $g(-t) = \bar g(t)$, the following are equivalent:
\begin{itemize}
\item[$(i)$]$V(s)H = H\, , \quad s\in \mathbb R$,
\item[$(ii)$] $g(A)H \subset H$, $g\in\C$,
\item[$(iii)$] $g(A)H \subset H$, $g\in\B$,
\item[$(iv)$]$(K^2 - 1)^{-1}H \subset H$ and $K(K^2 - 1)^{-1}H \subset H$,
\item[$(v)$]$(K \pm 1)^{-1}H \subset H$,
\item[$(vi)$] $K|_H$ is skew-selfadjoint on $H$, namely
 $D(K)\cap H$ is dense in $H$, $K(D(K)\cap H)\subset H$ and $K: (D(K)\cap H) \subset H\to H$
is skew-selfadjoint. 
\end{itemize}
\end{proposition}
\begin{proof}
$(i) \Rightarrow (v)$:
Let $h$ be a real $L^1$-function on $\RR$ and $g = \hat h$ its Fourier transform. Then
\[
g(A)  = \int_\RR h(s)V( s)ds \, .
\]
If $(i)$ holds true, then $g(A)H \subset H$ for such $g$. In particular $(v)$ holds, similarly as in \eqref{VH}. 

$(v)\Rightarrow (iv)$ follows by the identities 
\[
2K(K^2 - 1)^{-1} = (K + 1)^{-1} + (K - 1)^{-1}\, ,\quad 2(K^2 - 1)^{-1} = (K + 1)^{-1} - (K - 1)^{-1}\, .
\]

$(iv)\Rightarrow (ii)$: 
The complexification $\C_{\mathbb C}$ of $\C$ is the $C^*$-algebra $C$ of continuous functions on $\mathbb R$ vanishing at infinity, namely every $F\in C$ is written uniquely as
\[
F = f + i g\, , \quad f,g\in\C \ , \qquad f(t) =\frac12 [F(t) + \bar F(-t)],\ g(t) =\frac1{2i} [F(t) - \bar F(-t)].
\]
Note that
$\C$ contains the real algebra $\C_0$ of functions of the form $f(t) = p((t^2 +1)^{-1}, it(t^2 +1)^{-1})$ with $p$ a two-variables polynomial with real coefficients and zero constant coefficient. 
Note that $\C$ and $\C_0$ are closed with respect to complex conjugation. 

Since $\C_0$ separates the points of $\bR$, given any $f\in\C$, there exists a sequence $F_n = f_n + i g_n$ with $f_n, g_n\in \C_0$ such that $F_n \to f$ uniformly on $\bR$, thus $f_n \to f$ uniformly on $\bR$. 
We conclude that $\C_0$ is norm dense in $\C$. 
Therefore we have $(ii)$. 

$(ii)\Rightarrow (iii)$: Let $g\in \B$. Given $h\in H,\, k\in H^\perp$ (real orthogonal), by Lusin's theorem there exists a bounded sequence of continuous functions $g_n$ 
such that $g_n \to g$ almost everywhere w.r.t. the spectral measure of $A$ associated with $h,k$. By replacing $g_n(t)$ with $g_n(t) + \bar g_n(-t)$, we may assume that $g_n\in\C$. 
By Lebesgue's dominated convergence theorem, we have $(k, g_n(A)h) \to (k, g(A)h)$. As $\Re(k, g_n(A)h)=0$, we have $\Re(k, g(A)h)=0$, that implies $g(A)H\subset H$ because $h,k$ can be arbitrarily chosen. 

$(vi)\Rightarrow (v)$: $K|_H$ generates a one-parameter group $V_H$ of orthogonal operators on $H$, therefore
\begin{equation}\label{VH}
(1 \pm K)^{-1} H =  -\int_0^\infty e^{- s}V_H(\mp s)ds\, H \subset H\, .
\end{equation}

Clearly, $(i)\Rightarrow (vi)$, $(iii)\Rightarrow (ii)$ and $(iii)\Rightarrow (i)$. 
\end{proof}

\subsection{Modular Hamiltonian and cutting projections}\label{BW}

Let $\H$ be a complex Hilbert space. 
A {\it standard subspace} $H$ of $\H$ is a closed, real linear subspace of $\H$ with
\[
\overline{H + i H} = \H\, ,\quad H\cap i H = \{0\}\, .
\]
Let $H\subset\H$ be a standard subspace of the complex Hilbert space $\H$
and $\Delta_H$, $J_H$ be the {\it modular operator} and {\it conjugation} of $H$, namely
\[
S_H = J_H \Delta^{1/2}_H
\]
is the polar decomposition of the antilinear involution $S_H : h + i k \mapsto h -i k$, $h,k\in H$, ({\it Tomita's operator}). Then 
\[
\Delta_H^{is} H = H\, , \quad J_H H = H'\, ,
\]
$s\in \RR$, where $H'= iH^\perp$ is the {\it symplectic complement} of $H$, with $\perp$ denoting the orthogonal w.r.t. the real part of the scalar product. $\Delta_H^{is}$  and $\log\Delta_H$ are called the {\it modular unitary group} and the {\it modular Hamiltonian}  of $H$. 

In this section, and in the following one, we assume $H$ to be {\it factorial}, i.e. $H\cap H' = \{0\}$, that is $H + H'$ is dense in $\H$;
equivalently $1$ is not an eigenvalue of $\Delta_H$.  

Associated with $H$, there are three other standard subspaces: the real orthogonal $H^\perp$, the symplectic complement $H'$, and $iH$. We shall consider the three associated real linear projections
\begin{align*}
&E_H : H + H^\perp \to H\ , \quad h + h^\perp \mapsto h\ ,\\
&P_H : H + H' \to H\ , \quad  h + h' \mapsto h\ ,\\
&Q_H : H + iH\to H , \quad h + ik \mapsto h\ .
\end{align*}
Note that $E_H,P_H,Q_H$ are closed, densely defined real linear operators, $E_H$ is bounded and $P_H$ is the {\it cutting projection} \cite{CLR19}. We have
\begin{align}
&E_H = (1 + \Delta_H)^{-1} + J_H\Delta_H^{1/2} (1 + \Delta_H)^{-1}\ , \label{fE}\\
&P_H = (1 - \Delta_H)^{-1} + J_H\Delta_H^{1/2} (1 - \Delta_H)^{-1}\ , \label{fP} \\
&Q_H  =\frac12 ( 1 + S_H)\ ;  \label{fQ}
\end{align}
more precisely, \eqref{fP} means that
$P_H$ is the closure of the operator $a(\Delta_H) + b(\Delta_H)J_H$ on $ D(a(\Delta_H))\cap D(b(\Delta_H))$, with
\[
a(\l) = (1 - \l)^{-1}\ , \quad b(\l) = \l^{1/2}(1 - \l)^{-1}\ , \ \l\in (0,+\infty)\setminus\{1\}\, .
\]
Formulas \eqref{fE} and \eqref{fP} were obtained respectively in \cite{FG} (see also \cite{LX-vN}) and in \cite{CLR19}; formula \eqref{fQ} is straightforward. We then have
\[
E_H = 2Q_H(1 + \Delta_H)^{-1}\ , \quad P_H = 2Q_H(1 - \Delta_H)^{-1}\ ,
\]
hence
\ben\label{PE}
P_H = E_H (1 + \Delta_H) (1 - \Delta_H)^{-1} = E_H\coth\Big(\frac12\log\Delta_H\Big)
\een
(denoting by the same symbol an operator and its closure). 
\begin{proposition}\label{Pir}
$P_H i$ restricts to a skew-selfadjoint operator on $H$ and
\ben\label{gP}
P_H i|_H  = i\coth\Big(\frac12\log\Delta_H\Big)\Big|_H\, .
\een
\end{proposition}
\begin{proof}
$H$ is $i\tanh(\frac12\log\Delta_H)$-invariant by Prop. \ref{inv}, because the hyperbolic tangent is a bounded, odd function; so   $i\tanh(\frac12\log\Delta_H)|_H$ is bounded, skew-selfadjoint on $H$. 
Therefore, its inverse $-i\coth(\frac12\log\Delta_H)|_H$ is a skew-selfadjoint operator on $H$. Namely, $H$ is $i\coth(\frac12\log\Delta_H)$-invariant. 
Equation \eqref{PE} then gives \eqref{gP} by restriction. 
\end{proof}
We end this subsection by noting the following relations
\begin{gather}
P_H^* = -iP_H i\, ;\label{PP*}\\
 P_H i \log\Delta_H = i\log\Delta_H\, P_H \, .\label{PP**}
\end{gather}
The first relation is proved in \cite{CLR19}, the second one is valid on $D(P_H i \log\Delta_H)\cap D(i\log\Delta_H\, P_H)$
because $H$ is $i\log\Delta_H$-invariant. 

\subsection{Entropy and quadratic forms}\label{Eqf}
Recall the {\it entropy of a vector} $k$ in a Hilbert space $\H$ with respect to a factorial standard subspace $H$ is defined by
\ben\label{Se}
S_k = \Im(k, P_H i \log\Delta_H\, k) = \Re (k, P^*_H \log\Delta_H\, k)  \, ,
\een
where $P_H$ is the cutting projection $P_H : H + H' \to H$. Formula \eqref{Se} is to be understood in the sense of quadratic forms; indeed, we now define the quadratic form $S_k = S(k,k)$. 

Let $\D = D(\sqrt{|\log\Delta_H|}\, F)$, with $F$ the spectral projection of $\Delta_H$ relative to the interval $(0,1)$. Then $\D$ is a dense linear subspace of $\H$. 
With $h, k\in \D$, we set
\begin{equation}\label{Sb}
S(h, k) = \Re(h, P^*_H \log\Delta_H\, k) = \Im(h, P_H i\log\Delta_H k)   \, .
\end{equation}
More precisely, following the discussion il \cite{CLR19}, let $E(\l)$ be the spectral family of $\Delta_H$, 
the right-hand side of \eqref{Sb} is defined by on $\D$ by
\begin{equation}\label{Sb2}
\Im(h, P_H i\log\Delta_H k) =  \int_0^{+\infty}  a(\l)\log\l\, d(h,E(\l)k) - \int_0^{+\infty}  b(\l)\log\l \, d(h,J_H E(\l)k) \ .
\end{equation}
Note that $b(\l)\log\l$ is a bounded function, so the right integral is always finite. Moreover, $a(\l)\log\l$ is bounded on $(1 , +\infty)$ and positive on $(0,1)$. So the above formula is well defined by the spectral theorem, provided $h,k\in\D$.  In particular,
\[
S_k < \infty\ \Longleftrightarrow\ k\in \D\, ,
\]
\cite[Prop. 2.4]{CLR19}. 

Set $\D_0 =  D(a(\Delta_H)\log\Delta_H)\cap D(b(\Delta_H)\log\Delta_H) \subset D(\log\Delta_H\, P_H)\cap \D$. 

\begin{lemma}\label{PHself}
 $P^*_H \log\Delta_H$ is essentially selfadjoint on $\D_0$ and   $P^*_H \log\Delta_H |_{\D_0} = \log\Delta_H\, P_H |_{\D_0}$. 
 \end{lemma}
\begin{proof}
By equations \eqref{PP*}, we have
\[
(P^*_H  \log\Delta_H )^* h=  \log\Delta_H\, P _H h = -i i\log\Delta_H \, P _H h
= - i P _H i \log\Delta_H   h = P^* _H \log\Delta_H h
\]
for all $h\in \D_0$, namely $(P^*_H  \log\Delta_H)^*|_{\D_0} =  P^*_H  \log\Delta_H |_{\D_0} = \log\Delta_H\, P_H |_{\D_0}$,
so it suffices to show that the closure $\overline{\log\Delta_H\, P_H}$ of $\log\Delta_H\, P_H |_{\D_0}$ is  selfadjoint. 

With $0 < \e <1$, let $E_\e$ be the spectral projection of $\log\Delta_H$ relative to the set $(-\e^{-1}, -\e )\cup (\e, \e^{-1})$. Then $E_\e (H + H') \subset H + H'$ by using Prop. \ref{inv}; moreover  $E_\e (H + H') \subset \D_0$. 
So $E_\e \H \subset D(\overline{\log\Delta_H\, P_H})$ because 
$\log\Delta_H\, P_H\, E_\e$ is bounded by formula \eqref{fP} and $E_\e$ commutes with $\log\Delta_H\, P_H$, thus also with $(\log\Delta_H\, P_H)^*$. 

We conclude that $(\log\Delta_H\, P_H)^* E_\e = \overline{\log\Delta_H\, P_H} E_\e$ and
the lemma follows since $E_\e \nearrow 1$ as $\e\searrow 0$. 
\end{proof}
\begin{proposition}\label{entq}
 $S$ is a real linear, closed, symmetric, positive quadratic form on $\D$. 
 \end{proposition}
\begin{proof}
$S$ is real linear and positive.  By Lemma \ref{PHself}, $S$ is also symmetric on $\D$. 
Thus $S$ closable on $\D_0$, being associated with the closable, real linear operator $P^*_H \log\Delta_H$ by eq. \eqref{Sb}, 
cf. the proof of  \cite[Thm. 1.27]{Ka}.  Indeed, $S$ is closed on $\D$ because $\D_0$ is a core for $\overline{\log\Delta_H\, P_H}$, so a form core for $S$. 
\end{proof}

\subsection{Standard subspaces and passivity}\label{PW}
We provide here a standard subspace version of Pusz-Woronowicz's complete passivity characterisation of the modular Hamiltonian in the $C^*$-algebraic setting \cite{PW}. 

Let $\H$ be a complex Hilbert space and $H$ standard subspace.  
In this section, $A$ is a selfadjoint linear operator on $\H$ such that $e^{is A}H = H$, $s\in\mathbb R$, namely $H$ is $iA$-invariant as in Prop. \ref{inv}.

We shall say that $A$ is {\it active/passive} with respect to $H$ if
\[
\pm(\xi, A \xi) \geq 0\, ,\quad \xi\in D(A)\cap H\, .
\]
$A$ is {\it $n$-active/passive} w.r.t. $H$ if the generator of $e^{itA}\otimes e^{itA}\cdots \otimes e^{itA}$ is active/passive with respect to the $n$-fold tensor product $H\otimes H\otimes\cdots \otimes H$,  (closed real linear span of monomials $h_1\otimes h_2\otimes\cdots \otimes h_n$, $h_i\in H$,
cf. \cite{LMR16}). 
$A$ is {\it completely active/passive} if $A$ is $n$-active/passive for all $n\in\mathbb N$. 

Note that $A$ is active/passive iff 
\[
\pm(\xi, A \xi) \geq 0\, ,\quad \xi\in \cD \, ,
\]
with $\cD\subset D(A)\cap H$ a real linear space such that the closure of $\cD$ in the graph topology of $A$ is equal to $D(A)\cap H$. 

Since $e^{itA}$ leaves $H$ globally invariant, $e^{itA}$ commutes with $\Delta_H^{is}$ and $J_H$. We have 
\ben\label{sym}
J_H\Delta_H J_H = \Delta_H^{-1}\, \quad  J_H AJ_H = -A\, .
\een
 \begin{proposition}
$\log \Delta_H$ is completely passive w.r.t. $H$. 
 \end{proposition}
\begin{proof}
$\log \Delta_H$ is passive w.r.t. $H$, see \cite{L19}. Hence it is completely passive because the modular unitary group of 
$H\otimes H\cdots \otimes H$ is $\Delta_H^{is}\otimes \Delta_H^{is}\cdots \otimes \Delta_H^{is}$, see \cite{LMR16}. 
\end{proof}
 Let $\D_{\rm an} \subset\H$ be the  subspace of vectors with bounded spectrum with respect to both $A$ and $\log\Delta$.
\begin{lemma}\label{pos} If $A$ is passive w.r.t. $H$, then $A\log\Delta_H$ is a positive selfadjoint operator on $\H$. 
\end{lemma}
\begin{proof}
Denote by $S_H$ the Tomita operator of $\H$. 
As $S_H$ commutes with $e^{itA}$, $\Delta_H^{is}$ and $J_H$, we have $S_H \D_{\rm an}  = \D_{\rm an}$.
Let $\xi \in\D_{\rm an}$, so $(1+S_H)\xi\in H$. Then, by passivity,
\ben\label{S1}
0\geq ((1 + S_H)\xi, A(1+ S_H)\xi ) \ ,
\een
thus
\ben\label{S2}
0\geq ((1 + S_H)i\xi, A(1+ S_H)i\xi ) = ((1 - S_H)\xi, A(1- S_H)\xi )\ .
\een
Summing up \eqref{S1} and \eqref{S2} we get
\begin{multline*}
0\geq  (\xi,A\xi) + (S_H\xi, A S_H\xi) = (\xi,A \xi) + (J_H\Delta_H^{1/2} \xi,  AJ_H\Delta_H^{1/2} \xi) 
= (\xi, A\xi) - (J_H\Delta_H^{1/2} \xi,  J_HA\Delta_H^{1/2} \xi)  \\
 =  (\xi,A \xi) - (A\Delta_H^{1/2} \xi,  \Delta_H^{1/2} \xi)  = (\xi, A(1-\Delta_H)\xi)\, ,
 \end{multline*}
thus $A(1-\Delta_H) \leq 0$ because $\D_{\rm an}$ is a core for $A(1-\Delta_H)$. 
But $A(1-\Delta_H) \leq 0$ is equivalent to $A\log\Delta_H \geq 0$.  
\end{proof}
We shall say that a standard subspace is {\it abelian} if $\Delta_H =1$, see \cite{LN}. 
\begin{theorem}\label{Cpass}
$A$ is completely active with respect to $H$ iff $\log\Delta_H = \l A$ for some $\l \leq 0$. 
\end{theorem}
\begin{proof}
Assume that $A$ is completely active with respect to $H$. 
Let $\Lambda\subset \mathbb R^2$ be the joint spectrum of $A$ and $\log\Delta_H$.  
By Lemma \ref{pos}, $\Lambda$ is contained in the region $Q=\{ (a,b)\in\mathbb R^2: \, ab \geq 0\}$. By \eqref{sym} we have $-\Lambda = \Lambda$ and by
complete passivity we have $\Lambda + \Lambda + \cdots +\Lambda\subset Q$ (finite sum). 

Let $C_1, C_2$ be two different points in $\Lambda$. Then $n_1 C_1 + n_2 C_2\in Q$ for all integers $n_1,n_2\in\Z$. 
So $C_1$ and $C_2$ must belong to a same straight line $Z$ through the origin $(0,0)$. Thus $\Lambda\subset Z$. 
If $Z$ is  vertical, then  $\log\Delta_H = 0$, namely $H$ is abelian, so $\log\Delta_H = \l A$ with $\l=0$. 
If $Z$ is not vertical then $A = \l \log\Delta_H$ with $\l > 0$.  

For the converse, it remains to show that $A$ is completely active if $\Delta_H =1$, namely if $H$ is abelian. In this case, the scalar product of $\H$ is real on $H$, and $iH$ is the real orthogonal complement of $H$. 
As $A$ maps $H$ into $iH$, we have $\Re(h, A h) =0$ for $h\in D(A)\cap H$, thus $(h, A h) =0$ because $A$ is selfadjoint. So $A$ is active, thus completely active by repeating this argument for $H\otimes H\cdots \otimes H$. 
\end{proof}

\section{Preliminaries on the waves' space}
\label{pre}
Let $\S$ denote the real linear space of smooth, compactly supported real functions on $\mathbb R^d$. We shall always assume $d\geq 2$ unless otherwise specified. 

As is known, if $f, g\in\S$, there is a unique smooth real function $\Phi$ on $\mathbb R^{d+1}$ which is a solution $\Phi$ of the wave equation 
\[
\square\Phi \equiv \partial^2 \Phi/\partial x_0^2 - \partial^2 \Phi/\partial x_1^2 \cdots - \partial^2 \Phi/\partial x_d^2= 0
\]
(a wave packet or, briefly, a wave) with Cauchy data $\Phi |_{x_0 =0} = f$, $\partial_0\Phi |_{x_0 =0} = g$. 
We set $\Phi=w(f,g)$ and denote
by $\T$ the real linear space of these $\Phi$'s; we will often use the identification
\ben\label{st}
\S^2 \longleftrightarrow \T\, , \qquad \langle f,g\rangle \longleftrightarrow w(f, g)\, .
\een  
The one particle Hilbert space is $\H= L^2\big( {\mathfrak H}_0 , \delta(p^2)\big)$
with $\mathfrak H_0$ the positive massless hyperboloid, namely $\mathfrak H_0$ is the boundary of the forward light cone.  We denote by $(\cdot,\cdot)$ the scalar product of $\H$. The Fourier transform of a $\Phi \in \T$ is a distribution of the form $\hat \Phi(p) = \delta(p^2) F(p)$ with a compactly supported smooth function $F : {\mathfrak H}_0 \cup (-{\mathfrak H}_0)\to \bC$, so that
$\T$ real linearly embeds into $\H$ by 
$\Phi \mapsto (2\pi)^{d/2}F |_{\mathfrak H_0}$. We may thus consider $\T$ as a dense subset of $\H$.

Consider the symplectic form on $\T$
\ben\label{beta}
\b(\Phi,\Psi) = \frac12\int_{x_0 =0} \big(\Psi\partial_0 \Phi - \Phi\partial_0 \Psi) dx\,.
\een
This is the imaginary part of the restriction of the scalar product of $\H$ to $\T$:
\[
\Im(\langle f_1,g_1\rangle , \langle f_2 ,g_2 \rangle) = \b(\langle f_1,g_1\rangle , \langle f_2 ,g_2 \rangle)\ .
\]
We denote by $H_0^{s}$ the  real Hilbert space of real-valued tempered distributions $f \in S'(\bR^d)$ such that 
\ben\label{H12}
||f ||_{s}^2  = \int_{\RR^d} {|\bp|^{2s}} | \hat f(\bp)|^2 d\bp< +\infty\, ,\quad s\in\RR\, .
\een
It is clear that $\S$ is dense in $H^{\pm 1/2}$ and that $\mu : H_0^{1/2} \to H^{-1/2}$ with
\ben\label{mum}
\widehat{\mu f}({\bm p}) = {{|\bm p|}}\hat f({\bm p})\, .
\een
is a unitary operator. 

Then
\ben\label{imum}
\imath_0 = \left[\begin{matrix}
0 & \mu^{-1} \\ -\mu  & 0  
\end{matrix}\right]\, ,
\een
namely $\imath_0\langle f,g\rangle = \langle \mu^{-1} g,  -\mu f\rangle$, 
is a unitary operator $\imath_0$ on $H =  H^{1/2}\oplus H^{-1/2}$.

As $\imath_0^2 = -1$,  the unitary $\imath_0$ defines a complex structure (multiplication by the imaginary unit)  on
$H_0$ that becomes a complex Hilbert space with scalar product
\[
(\Phi,\Psi) = \b(\Phi, \imath_0 \Psi) + i\b(\Phi, \Psi) \ .
\]
In terms of the Cauchy data, the scalar product is given by
\[
( \langle f,g\rangle, \langle h,k \rangle ) = \frac1 2\big((f,\mu h)+(g,\mu^{-1}k) + i [(h, g)-(f,k)]\big), \qquad \langle f,g\rangle,\langle h,k\rangle \in H^{1/2}\oplus H^{-1/2},
\]
(where $(\cdot, \cdot)$ is the standard $L^2$ scalar product).
Furthermore the map
\[
\langle f,g\rangle \in H \mapsto w(f,g) \in \H
\]
is isometric and complex linear, so it extends to a unitary operator. We will then identify these spaces; namely, the one particle Hilbert space is 
\[
\H = \text{real Hilbert space}\  H \ \text{with complex structure given by} \ \imath_0\, .
\]
Note that
$H^{1/2}$ and $H^{- 1/2}$ are naturally a dual pair under the bilinear form
\ben\label{dual}
\langle f,g\rangle \in H^{1/2}\times H^{-1/2} \mapsto \int_{\bR^d} \bar{\hat f}\hat g\ .
\een
Let $Z$ be an open, non-empty subset of $\mathbb R^{d}$ and denote by $Z'$ the interior of its complement. 
We shall denote by $H(Z)$ the closed, real linear subspace of $\H$
\[
H(Z) = \{\Phi\in\T(Z)\}^-\, ,
\]
where
\[
\T(Z) = \{\Phi = w(f,g)\in\T : {\rm supp}(f),\, {\rm supp}(g)\subset Z\}\, .
\]
If $Z$ and $Z'$ are non-empty, then $H(Z)$ is a standard subspace of $\H$, the one associated with $Z$. We shall be mainly interested in the case $Z$ or $Z'$ is the unit ball $B$ of $ \mathbb R^d$. Clearly, setting
\[
H^{\pm 1/2}(Z) = \text{closure of $C_0^\infty(Z)$ in $H^{\pm 1/2}$}\, ,
\]
we have
\[
H(Z) = H^{1/2}(Z)\oplus H^{-1/2}(Z)\, .
\]
Here and in the following, $C^{\infty}_0(Z)$ denotes the space of real $C^\infty$ function on $\mathbb R^d$ with compact support in $Z$.
By duality \cite{Ar63}, see also \cite{EO,LRT}, we have
\begin{equation}\label{duality}
H(B)' = H(B')\, .
\end{equation}

\section{Massless Hamiltonian}\label{MasslessH}
We now compute the formula for the massless Hamiltonian. 
\subsection{The modular group} 
In the following, $O$ will denote the double cone on the Minkowski spacetime $\mathbb R^{d+1}$ with base the open unit ball $B$ centered at the origin in the time zero hyperplane $\mathbb R^d$. As above,  $H(O)=H(B)$ is the standard subspace in the massless, scalar one-particle Hilbert space $\H$. We assume $d\geq 2$. 

The modular group $\Delta_B^{is}$ associated with $H(B)$ has been computed in \cite{HL} in terms of the action on the field or, equivalently, on the spacetime test functions. $\Delta_B^{is}$ is associated with a one-parameter group of conformal transformation that preserves $O$, namely
\ben\label{Z}
(u,v) \mapsto \big((Z(u,s), Z(v,s)\big)\, ,
\een
where $Z$ is given by
\[
Z(z,s) = \frac{g(z, s)}{f(z, s)}
\]
with
\[
f(z, s) = \frac{(1+ z) + e^{-s}(1-z)}{2}  \ ,\quad g(z, s) = \frac{(1+ z) - e^{-s}(1- z)}{2} \ .
\]
Here we need to compute $\Delta_B^{is}$ in terms of waves.   
Let $\Phi$ be a wave and set
\ben\label{V}
(V(s)\Phi)(u, v) = \g(u,v;s)\Phi\big(Z(u,s), Z(v,s)\big)\, ,
\een
with  
\[
 u = x_0 + r, \quad v = x_0 - r, \quad r = |\bx| \equiv \sqrt{x_1^2 + \cdots+ x_d^2}
\]
and we omit the remaining spherical coordinates as the action is trivial on them. 

The cocycle $\g$ given by
\[
\g(u,v;s) = F(u, s)F(-v, -s), \quad F(z,s) \equiv f^{-D}(z,s) \ ,
\]
with
\ben\label{D}
D = \frac{d-1}{2}\, ;
\een
$D$ is the scaling dimension. 
We have:
\begin{theorem}\label{MG}
The modular group of $H(B)$ is given by
\[
\Delta_B^{-is} = V(2\pi s)\, .
\]
\end{theorem}
\begin{proof}
That $\Delta_B^{-is/2\pi}$ is associated with the flow $Z$ \eqref{Z} and the cocycle $\g$ follows rather directly from the test function formula \cite{HL}, but for the determination of the value of the constant $D$. Now, as in \cite{HL}, $H(B)$ is equivalent to $H(W)$, with $W$ a wedge region, by a unitary operator that is obtained by composing unitaries associated with translations and ray inversion map. The cocycle is associated with the Jacobian of the ray inversion map. Thus, by computation similar to those in \cite{HL}, the cocycle is of the form $f^{-D}$ for some constant $D>0$. 

However, it will follow from Lemma \ref{Ki} that $D = (d-1)/2$ is the only value compatible with the complex linearity of the generator of $V$. 
Of course, $D$ can be determined by direct calculations too. 
\end{proof}

\subsection{The generator of the modular group}
We now compute the $K = \frac{d}{ds}V(s)\big|_{s=0}$, the generator of $V$ in Theorem \ref{MG}, that is proportional to the modular Hamiltonian. 

With the above notations, denoting by a prime the derivative with respect to $s$-parameter, we have (Identity \ref{zg})
\[
Z(z,0) = z\, , \qquad Z'(z,0) = (1-z^2)/2 
\]
and
\[
\g(u,v; 0) = 1\, , \qquad\g'(u,v;0) = -\frac{D}2(u +v) = -D\,x_0\ .
\]
Therefore
\begin{multline*}
(V(s)\Phi)(u,v)'  = \g'(u,v;s)\Phi\big(Z(u,s), Z(v,s)\big) \\
+  \g(u,v;s)\Big(\partial_u \Phi\big(Z(u,s), Z(v,s)\big)Z'(u,s) +
\partial_v \Phi\big(Z(u,s), Z(v,s)\big)Z'(v,s)\Big)\, .
\end{multline*}
We then have:
\begin{proposition}We have
\[
(K \Phi)(u,v) =  -\frac{D}2(u +v)\Phi 
 -\frac12 u^2\partial_u \Phi -\frac12 v^2\partial_v \Phi 
 + \frac12 \partial_u \Phi + \frac12 \partial_v \Phi\ .
\]
In terms of the $x_0, r$ coordinates
\ben\label{K}
(K \Phi)(x_0,r)
= \frac12\big(1 -(x^2_0 + r^2)\big) \partial_0 \Phi - x_0r \partial_r\Phi - D\, x_0\Phi \ .
\een
\end{proposition}
\begin{proof}
We compute:
\begin{align*}
(K \Phi)(u,v) & = (V(s)\Phi)(u,v)'\big|_{s =0} \\  & = \g'(u,v;0)\Phi(u, v)  + \partial_u \Phi(u, v)Z'(u,0) +
\partial_v \Phi(u, v)Z'(v,0)\\
& =  -\frac{D}2(u +v)\Phi(u,v) +
 \frac12\partial_u \Phi(u, v)(1-u^2) +
\frac12\partial_v \Phi(u, v)(1-v^2) \\
& =  -\frac{D}2(u +v)\Phi 
 -\frac12 u^2\partial_u \Phi -\frac12 v^2\partial_v \Phi 
 + \frac12 \partial_u \Phi + \frac12 \partial_v \Phi\ .
\end{align*}
Now
\[
\partial_u = \frac12(\partial_0  + \partial_{r}),\quad \partial_v = \frac12(\partial_0  - \partial_{r})\, ,
\]
so
\begin{multline}
(K \Phi)(x_0,r)
=   -\frac14\big((x_0 + r)^2 (\partial_{0} + \partial_{r})\Phi  +(x_0 -r)^2 (\partial_0  - \partial_{r})\big) 
\Phi -{D}x_0\Phi + \frac12\partial_0\Phi
\\
=  -\frac12(x^2_0 + r^2) \partial_0 \Phi - x_0r \partial_r\Phi -{D}x_0\Phi + \frac12\partial_0\Phi \\
= \frac12\big(1 -(x^2_0 + r^2)\big) \partial_0 \Phi - x_0r \partial_r\Phi - D\,x_0\Phi \ .
\end{multline}
\end{proof}
\begin{corollary}\label{KK'}
We have
\begin{gather}
(K \Phi)|_{x_0 = 0}
= \frac12(1 - r^2)\partial_0 \Phi |_{x_0 =0} \, ,\\
(\partial_0  K \Phi)|_{x_0 = 0} = \frac12(1 - r^2)\nabla^2\Phi - r\partial_r\Phi  -D\Phi |_{x_0 = 0} \, . \label{K2}
\end{gather}
\end{corollary}
\begin{proof}
The first equality follows immediately from \eqref{K}. 
Again from \eqref{K} we have
\[
\partial_0  K \Phi|_{x_0 = 0}  =  \frac12(1 - r^2)\partial^2_{0}\Phi - r\partial_r\Phi  -D\Phi |_{x_0 = 0} 
\]
thus \eqref{K2} holds. 
\end{proof}
The above corollary translates into the following:
\begin{proposition}\label{Kprop}
We have
\[
K : w( f, g)   \mapsto w \Big( \frac12(1 - r^2)g,   \frac12(1 - r^2)\nabla^2 f - r\partial_r f -D\,f\Big)\, .
\]
In other words, $K$ is the operator on $\S^2$ given by the $2\times 2$ matrix
\ben\label{Kmat}
K  = \left[
\begin{matrix}
0 & \frac12(1 - r^2) \\ \frac12(1 - r^2)\nabla^2 - r\partial_r  - D & 0
\end{matrix}\right] \, .
\een
\end{proposition}
Now the symplectic form $\b$ on $\T$ \eqref{beta} is given, in terms of Cauchy data, as
\ben\label{beta2}
\b(\Phi,\Psi) = \frac12\int_{\mathbb R^d} \big(f_2 g_1 - f_1 g_2 \big) d\bx \ .
\een
with $\Phi = w (f_1, g_1)$, $\Psi = w(f_2, g_2)$.

Thus
\begin{multline}\label{b0}
2\b(\Phi, K \Phi) = \int_{x_0 =0}  (K \Phi)\partial_0 \Phi - \Phi\partial_0 (K \Phi)dx\\
= \int_{\mathbb R^d}\frac12(1-r^2)g^2d\bx
 - \int_{\mathbb R^d}\frac12(1-r^2)f\nabla^2 fd\bx 
+ \int_{\mathbb R^d}rf\partial_r f d\bx
+ D\int_{\mathbb R^d}f^2d\bx \ .
\end{multline}
So, taking into account Identity \eqref{r12}, we have
\[
2\b(\Phi, K \Phi) 
= \frac12\int_{\mathbb R^d}(1-r^2)\big(g^2 + |\nabla f|^2\big)d\bx 
+  D\int_{\mathbb R^d}f^2d\bx \ ,
\]
namely
\ben\label{b1}
\b(\Phi, K \Phi) 
= \frac14\int_{x_0 =0}(1-r^2)\big( (\partial_0 \Phi)^2 + |\nabla\Phi|^2 \big)dx  + \frac{D}2 \int_{x_0 =0}\Phi^2dx \, .
\een
So we have obtained the following proposition.
\begin{proposition}\label{KT}
With $A =- \imath_0 K$, we have
\ben\label{Af}
(\Phi,A \Phi) =
\b(\Phi,K \Phi) = \frac12\int_{x_0 =0}(1 -r^2)\langle T_{00}\rangle_\Phi dx  + \frac{D}2\int_{x_0 =0}\Phi^2dx
\een
with $\langle T_{00}\rangle_\Phi = \frac12\big((\partial_0 \Phi)^2 + |\nabla\Phi|^2\big)$ the energy density given by 
classical stress-energy tensor $T$. 
\end{proposition}
\begin{proof}
Since $A$ is selfadjoint, $(\Phi,A \Phi)$ is real, so
\[
(\Phi,A \Phi) = \Im i(\Phi,A \Phi) = \Im (\Phi, \imath_0 A \Phi) = \Im (\Phi, K \Phi) = \b(\Phi, K \Phi) \, ,
\]
so the proposition follows from \eqref{b1}. 
\end{proof}
Since
\[
\b(K \Phi, \Psi) + \b(\Phi, K\Psi) = 0\, ,
\]
by the polarisation identity, we then get
\begin{multline}\label{hs}
\b(\Phi, K \Psi) =
\frac12\Big(\b\big(\Phi + \Psi , K (\Phi +\Psi)\big) - \b\big(\Phi - \Psi , K (\Phi -\Psi)\big)\Big)\\
= \frac12\int_{x_0 = 0}(1-r^2) \langle T_{00}\rangle_{\Phi, \Psi} dx
+ \frac{D}2\int_{x_0 = 0}\Phi\Psi dx \, ,
\end{multline}
with
\ben\label{ed2}
\langle T_{00}\rangle_{\Phi, \Psi}  = \frac12 \big( \partial_0 \Phi \partial_0 \Psi  + \nabla_{\bx} \Phi \nabla_{\bx} \Psi \big)\ .
\een

\subsection{$K$ acting on $\H$}
Set $H(B)  = H^{1/2}(B) \oplus H^{-1/2}(B) $ for the standard subspace associated with $B$ as before.

Let 
\[
M: D(M)\subset H^{-1/2} \to H^{1/2}\, , \quad L : D(L)\subset H^{1/2} \to H^{-1/2}
\] 
be the closures of the operators
\begin{align*}
M &= \frac12  (1 - r^2)  \\
L  &= \frac12 (1-r^2)\nabla^2 - r \partial_r - D
\end{align*}
on $\S$. ($M$ is the multiplication operator by $\frac12  (1 - r^2)$). 
We shall indeed see that both operator $M$ and $L$ on $\S$ are closable as $M^*$ contains $L$ on $\S$. 

Denote by $K = -\frac{\imath_0}{2\pi} \log \Delta_B$ the modular Hamiltonian relative to $B$.
\begin{lemma}\label{essa}
$\S^2$ is a core for $K$ (as real linear operator).
\end{lemma}
\begin{proof}
By the geometrical action of the modular group $\Delta_B^{is}$, and duality, one sees that $C_0^{\infty}(B)^2 + C_0^{\infty}(B')^2$ is a dense,  real linear subspace of $H$, globally $\Delta_B^{is}$-invariant,  contained in the domain of the generator $K$; thus it is a core for $K$. So $\S^2$ is a core to being a larger subspace still contained in $D(K)$.   
\end{proof}
So we have:
\begin{theorem}\label{mainth}
The massless modular Hamiltonian $\log \Delta_B$ is given by $- 2\pi A = \log \Delta_B$ with 
\[
K = \left[\begin{matrix} 0 & M \\
L &0\end{matrix}\right] =  \left[\begin{matrix} 0 & M \\
-M^* &0\end{matrix}\right]\, .
\]
and $A = -\imath_0 K$. 
\end{theorem}
\begin{proof}
By Lemma \ref{essa}, we have the first equality. As $A$ is selfadjoint, we have $K^* = -K$, thus the second equality holds. 
\end{proof}
\begin{corollary}\label{Ki}
We have $K   = -\imath_0  K \imath_0$. Thus 
\[
\mu M \mu =  M^* = -L \, .
\]
\end{corollary}
\begin{proof}
Immediate by the complex linearity of $K$. 
\end{proof}
Note, in particular, that the constant $D$ in the expression of the operator $L_0$ is fixed by the above corollary.  
\begin{corollary}\label{inv0} 
$K^B = K |_{H(B)}$ is skew-selfadjoint on $H(B)$. 
\end{corollary}
\begin{proof}
Since $H(B)$ is globally invariant for $e^{sK}$, we may apply Prop. \ref{inv}. 
\end{proof}
\begin{remark}\label{RDP} By Cor. \ref{inv0}, we have indirectly solved a Dirichlet problem for the degenerate elliptic operators $L M - 1$ and $M L  - 1$, see \cite{Bau}. 
\end{remark}
The following corollary manifests the thermodynamical nature of the modular Hamiltonian, cf. \cite{PW}. 
\begin{corollary}
$A$ is completely active with respect to $H(B)$. 
\end{corollary}
\begin{proof}
Immediate by Theorem \ref{Cpass}. 
\end{proof}
Let's now consider balls of different radii. 
If $F$ is a function on $\mathbb R^n$ and $\l >0$, we set $F_\l(x) = F(\l x)$. 
If $\Phi\in\T$, we have
\[
(\square\Phi_\l)(x) 
= \l^2 (\square\Phi)(\l x)  = 0\, ,
\]
thus $\Phi_\l \in \T$ and we have a real linear map bijection $\delta_\l : \T \to \T$
\ben\label{SD}
\delta_\l : \Phi \in\T \to \l^{D}\Phi_{\l }\in \T\ .
\een
In terms of the Cauchy data, we have $
\delta_\l : w(f,g) \to w( \l^D f_\l , \l^{D+1} g_\l)$. 

If $\Phi,\Psi\in\T(B)$  
\[
\b(\delta_\l\Phi, \delta_\l\Psi) 
= \frac{\l^{d}}2\int_{x_0 =0} \Big(\Psi(\l x)\big(\partial_0 \Phi\big)(\l x) - \Phi(\l x)\big(\partial_0 \Psi\big)(\l x) \Big) dx =\b(\Phi , \Psi )\ ,
\]
namely $\delta_\l$ preserves the symplectic form $\b$. One can see that $\delta_\l$ commutes with the complex structure on $\imath_0$, so $\delta_\l$ is a unitary
\[
\delta_\l : \H \to \H \, .
\]
Let $B_R$ the ball in $\mathbb R^d$ with center at the origin and radius $R >0$.  With $H(B_R)$ the standard subspace of $\H$ associated with $B_R$, we have
\[
\delta_\l : H (B_R)\to H (B_{\l^{-1}R} ) \, .
\]
With $\Delta_{R}$ the modular operator of $H(B_R)$, we then have $\Delta_{{\l^{-1}R}} =  \delta_\l \Delta_{R}\delta_{\l^{-1}}$, so
\ben\label{cR1}
\log \Delta_{R} =  \delta_{R}\log \Delta_{1}\delta_{R^{-1}}\, .
\een
The above discussion yields the following formula for the matrix elements of the modular Hamiltonian. 
\begin{theorem}\label{ModR}
With $\log \Delta_{R}$ the massless modular Hamiltonian for the radius $R$ ball $B_R$, we have 
\[
-\Re (\Phi, \log \Delta_{R} \Psi) =
2\pi\int_{B_R} \frac{R^2 - r^2}{2R} \langle T_{00} \rangle_{\Phi, \Psi} d\bx + 2\pi \frac{d-1}{R}\int_{B_R} \Phi\Psi d\bx
\, ,
\]
where $\Phi, \Psi\in\T(B_R)$. 
\end{theorem}
\begin{proof}
Immediate by the above discussion. 
\end{proof}

\section{Entropy density of a wave packet}
With $\H$  our wave Hilbert space, the following proposition gives an explicit formula for the action of the cutting projection $P_B$ on $\H$ relative to $H(B)$, cf. \cite{FG}. 
\begin{proposition}\label{PmM}
$P_B$ is  given by the matrix 
\begin{equation}\label{Pm}
P_B = 
\left[\begin{matrix} P_+ & 0 \\ 
0  &P_-\end{matrix}\right]\, ,
\end{equation}
with $ P_\pm: D(P_\pm) \subset H^{\pm 1/2} \to H^{\pm 1/2}$ the operator of multiplication by the characteristic function $\chi_B$ of $B$ in $H^{\pm 1/2}$. 
\end{proposition}
\begin{proof}
By duality  \eqref{duality}, $H(B)'= H(B')$, therefore $D(P_B) = H(B) + H(B')$. 

By Prop. \ref{Pir}, $P_B \imath_0 |_{H(B)}$ is a real linear, densely defined operator on $H_0(B)$.
Clearly $P_B$ acts as in \eqref{Pm} on vectors in $H(B) + H(B')$ given by smooth functions, hence on all $D(P_B)$ because $P_B$ is a closed operator \cite{CLR19}. 
\end{proof}

\begin{proposition}\label{PH}
Let $\Phi\in\T$ and set $\Psi = K\Phi$. Then $\Psi =w(f,g)\in H(B) + H(B')$ and
$P_B \Psi = w(\chi_B f, \chi_B g)$, with $\chi_B$ the characteristic function of $B$.  
\end{proposition}
\begin{proof}
As $\Phi$ belongs to  $D(\log \Delta_B) = D(K)$, it follows that $K \Phi\in D(P_B)$, thus $P_B \Psi = w(\chi_B f, \chi_B g)$ by Prop. \ref{PmM}. 
\end{proof}
We can now compute the local entropy of a wave packet $\Phi$. 
\begin{proposition}\label{entr1}
Let $\Phi\in \T$.  The entropy $S_{\Phi}$ of  $\Phi$ with respect to   $H_0(B)$ is given by 
\ben\label{entr1f}
S_{\Phi} = \pi\int_{B} \big(\Psi\Phi'    - \Phi \Psi' \big)d\bx \  ,
\een
with $\Psi = K \Phi$ and the prime denotes the time derivative. 
\end{proposition}
\begin{proof}
The proof now follows from Proposition \ref{PH} in analogy with the proof given in \cite{CLR19} for the wedge region case. 
\end{proof}
\begin{theorem}
The entropy $S_\Phi$ of the wave $\Phi\in \T$ in the unit ball $B$, i.e. with respect to $H(B)$, is given by
\ben\label{Sform}
S_{\Phi} = 2\pi\int_{B} \frac{1 - r^2}{2} \langle T_{00}\rangle_{\Phi}\, d\bx 
+  \pi D\int_{B} \Phi^2 d\bx \, .
\een
In particular, $S_\Phi$ is finite for all $\Phi\in\T$. 
\end{theorem}
\begin{proof}
With $\Phi,\Psi \in \T$, define the real linear quadratic form 
\[
q_B(\Phi, \Psi) = 2\pi\int_{B} \frac{1 - r^2}{2} \langle T_{00}\rangle_{\Phi,\Psi}\, d\bx 
+  \pi D\int_{B} \Phi\Psi d\bx 
\, .
\] 
Clearly, $q_B(\Phi, \Phi) < \infty$ if $\Phi\in\T$. We want to show that
\begin{equation}\label{Sq}
S(\Phi,\Phi) =  q_B(\Phi, \Phi)\, ,\quad \Phi\in \T\, ,
\end{equation}
where $S$ is the entropy form with respect to $H(B)$.
As $S$ is a closed by Prop. \ref{entq}, it suffices to show that $q_B$ is closable on $\T$ and that \eqref{Sq} holds on a form core for $S$. 

Now, $q_B$ is given by 
\[
q_B(\Phi, \Psi) = \Re (\Phi, P^*_B A P_B  \Psi)\, , \quad \Phi, \Psi \in \T\, ,
\]
similarly as in \cite[Thm. 3.5]{CLR19}. The real linear operator $P^*_0 A P_B$ is Hermitian, thus closable. As $q_B$ is positive, $q_B$ is closable, cf. the proof of  \cite[Thm. 1.27]{Ka}. 

On the other hand, \eqref{Sq} holds if $\Phi\in\T(B)$ by Theorem \ref{ModR}. Then it holds if $\Phi\in D(\log\Delta_B)$ by the same argument, using Prop. \ref{entr1}.  As $D(\log\Delta_B)$ is a form core for $S$  (see Sect. \ref{Eqf}), we conclude that $S = q_B$ on $\T$. 
\end{proof}
Denote by $B_R(\bar\bx)$ the radius $R$ space ball around the point $\bar\bx\in\mathbb R^d$. 
\begin{corollary}\label{SR}
The entropy $S_\Phi(R) = S_\Phi(R, t,\bar\bx)$ of the wave packet $\Phi\in\T$ in the space region $B_R(\bar\bx)$ at time $t$ is given by
\[
S_\Phi(R) = \pi\int_{B_R(\bar\bx)} \frac{R^2 - r^2}{R} \langle T_{00}\rangle_{\Phi}\, d\bx 
+  \frac{\pi D}{R}\int_{B_R(\bar\bx)} \Phi^2 d\bx 
\, ,
\]
with $r = |\bx - \bar\bx|$, ($x_0 = t$ integral). 
\end{corollary}
\begin{proof}
In view of Theorem \ref{ModR} and the formula for the cutting projection in Proposition \ref{PH}, the theorem follows, in the time zero case, by the entropy formula \eqref{Se}. By translation covariance, we get the formula at an arbitrary time. 
\end{proof}
\begin{figure}
\centering
\begin{tikzpicture}[xscale=0.9]
\draw [help lines, ->] (0,0) -- (11,0);
\draw [help lines, ->] (0,0) -- (0,4.3);
\node at (1.6,2.7) {$S$};
\node[blue] at (0.6,1.6) {$\Phi$};
\node[red] at (1.1,0.6) {$\Phi$};
\node at (11,-0.3) {$r$};
\node at (0.4,3.3) {$\frac1r$};
\draw[fill] (1,0) circle [radius=0.7pt]; 
\draw[fill] (0,0) circle [radius=0.7pt]; 
\node at (0,-0.3) {$0$};
\node at (1,-0.3) {$1$};
\draw [blue,smooth,very thick, domain= 0.01:3*pi] plot (\x, {2*(sin(\x r)* (1/\x)});
\draw [red,smooth,  very thick,domain= 1:3*pi + 1] plot (\x, {2*(sin((\x - 1)r)* (1/\x)});
\draw [blue,smooth,thick, domain= 0.95:3*pi] 
plot (\x, {4*pow(sin(\x r)* (1/\x),2)  
+   2*pow({(cos(\x r)*\x - sin(\x))* (1/(\x*\x)},2)
+ 2*pow(cos(\x r)*(1/\x),2)});
\draw [dashed,red,smooth, ultra thick,domain= 1:3*pi + 1] 
plot (\x, {4*pow{(sin((\x-1) r)* (1/\x)},2)  
+   2*pow({(cos((\x-1) r)*\x - sin(\x - 1))* (1/(\x*\x)},2)
+ 2*pow(cos((\x - 1) r)*(1/\x),2)
});
\draw[dotted,smooth, very  thick, domain = 0.46:11] plot(\x, 2/\x);
\end{tikzpicture}
\caption{The spherical massless wave $\Phi(r) = \frac{\sin (r-t)}{r}$, $d=3$, at time $t=0$ and $t=1$ ($\Phi$ is not everywhere defined).
In this example, the entropy densities per area $S(r)$ are close at different times.}
\end{figure}
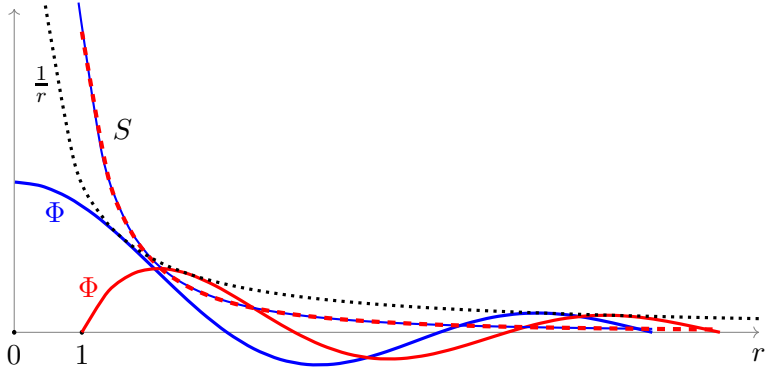
For large $R$, $S_\Phi(R)/R$ gets proportional to the total local energy 
$E= \int_{B_R(\bar\bx)} \langle T_{00}\rangle_{\Phi}(t,\bx) d\bx$; that is, since the restriction of $\Phi$ to the $x_0  = t$ hyperplane is compactly supported, we have
\[
 \frac{S_\Phi(R)}{R} \sim \pi E \, ,
\]
as $R\to\infty$ like in \cite[(41)]{L18}, see also \cite{LXbek}. This is in agreement with the Bekenstein bound
\[
S_\Phi(R) \leq \pi E R\, ,
\]
see \cite{LXbek} and references therein. 

On the other hand, at fixed time $t$, the local entropy $S_\Phi(R,\bx) = S_\Phi(R,t, \bx)$ has an expansion as $ R\to 0$ given by
\begin{equation}\label{asymp}
S_\Phi(R,\bx) =  \pi   D\Phi^2(t,\bx) V_d\, R^{d-1} + ...  \ =  \pi\frac{D}{d} A_{d-1}(R) \Phi^2(t,\bx) + ...\ ;
\end{equation}
here $ A_{d-1}(R)  = 2\frac{\pi^{d/2}}{\Gamma(d/2)}R^{d-1}$ is the area of the $(d-1)$-dimensional sphere $\partial B_R$ and  $V_d = A_{d-1}(1)/d$ is the volume of $B$. 
So the entropy density of the wave packet $\Phi$ around a point gets proportional to the area of the sphere boundary of $B_R$, as expected by  holographic area theorems for the entropy, in a black hole and other contexts, see \cite{Bek}. 

Note that the ratio entropy-density/area $S_\Phi(R,x)/A_{d-1}(R)$ in \eqref{asymp} is  proportional to the height $\Phi^2(t,\bx)$ of the wave packet at the point $x =(t,\bx)$, that may be interpreted as a (non normalised) probability density according to the Born rule. 

\subsection{Quantum Field Theory}
By the analysis in \cite{L19,CLR19}, we have an immediate corollary in Quantum Field Theory concerning the local vacuum relative entropy of a coherent state. 
\begin{corollary}
Let $\A(O_R)$ be the von Neumann algebra associated with the double cone $O_R$ (the causal envelope of $B_R$) by the free, neutral, massless quantum field theory. The relative entropy $S(\f_\Phi |\!| \f)$ (see \cite{Ar76}) between the vacuum state $\f$ and the coherent state $\f_\Phi$ associated with the one-particle wave $\Phi\in\H$ is given by $S_\Phi(R)$ by Corollary \ref{SR} (with $B_R$ centred at the origin).
\end{corollary}
\begin{proof}
As shown in \cite{CLR19}, $S(\f_\Phi |\!| \f)$ is equal to the entropy of the vector $\Phi\in \H$ with respect to the standard subspace $H(O_R)$. So Corollary \ref{SR} applies. 
\end{proof}

\section{Appendix. Elementary relations}
For the reader's convenience, we collect a couple of elementary identities that are used in the text. 
Note first that $\partial_r$ is the partial derivative in the direction 
$\vec{r} = (\frac{x_1}{r},\dots \frac{x_d}r)$, thus
\ben\label{dr}
r\partial_r f = {\bm x}\cdot{ \nabla}f \, ,
\een
for any $f \in \cS$.
\begin{identity}\label{r12}
Let $f\in\S$. We have
\[
\int_{\mathbb R^d}\frac12(1-r^2)|\nabla f |^2  d\bx  =
-\int_{\mathbb R^d}\frac12(1-r^2)f \nabla^2 f   d\bx  +
  \int_{\mathbb R^d} rf\partial_r f\, d\bx
  \, .
\]
\end{identity}
\begin{proof}
The identity follows immediately by the following two relations:
\ben\label{r1}
-\int_{\mathbb R^d}\frac12(1-r^2)f \nabla^2 f   d\bx  = 
  \int_{\mathbb R^d}\frac12(1-r^2)|\nabla f |^2  d\bx 
 +\frac{d}2\int_{\mathbb R^d} f ^2 d\bx\ ,
\een
\ben\label{r2}
\int_{\mathbb R^d} rf\partial_r f\, d\bx  = - \frac{d}2\int_{\mathbb R^d}  f^2  d\bx \ .
\een
Concerning the second relation, by \eqref{dr} we have
\begin{multline*}
\int_{\mathbb R^d} rf\partial_r f\, d\bx  = 
\sum_k \int_{\mathbb R^d} x_k f\partial_k f\, d\bx  =
\frac12\sum_k \int_{\mathbb R^d} x_k \partial_k (f^2) d\bx \\
= -\frac12\sum_k \int_{\mathbb R^d} f^2 d\bx
= - \frac{d}2\int_{\mathbb R^d}  f^2  d\bx \ .
\end{multline*}
Then the first relation follows by
\begin{align*} 
\int_{\mathbb R^d}(1-r^2)f \nabla^2 f   d\bx & = -\sum_k \int_{\mathbb R^d}\partial_k\big((1-r^2)f \big)\partial_kf \,  d\bx \\ 
& = 2\sum_k \int_{\mathbb R^d}x_kf \partial_k  f   \, d\bx
 -\sum_k \int_{\mathbb R^d}(1-r^2)(\partial_k f )^2  d\bx
\\
&  = -d\int_{\mathbb R^d} f ^2 d\bx
 - \int_{\mathbb R^d}(1-r^2)|\nabla f |^2  d\bx \ .
\end{align*}
\end{proof}
\begin{identity}\label{zg}
With $Z$ and $\g$ as in Section \ref{pre}, we have
\[
Z'(z,s)|_{s=0} = (1-z^2)/2 \, ,
\]
\[
\g'(u,v;0) = -\frac{D}2(u +v) = -D\,x_0\, .
\]
\end{identity}
\begin{proof}
Denoting by a prime the derivative with respect to the $s$-parameter, we have
\[
f(z,0) = 1,\quad f'(z, 0) = -\frac{e^{-s}(1-z)}{2}\big|_{s=0} = -\frac{1-z}{2}\, ,
\]
\[
g(z,0) = z,\quad g'(z, s) = \frac{e^{-s}(1-z)}{2}|_{s=0} = \frac{1-z}{2}\, .
\]
Since $Z = g/f$, we get
\[
Z'(z,s)|_{s=0} = (1-z^2)/2 \, .
\]
We have
\[
\g'(u,v;s)\big |_{s=0} = \big(F(u, s)F(-v, -s)\big)' \big |_{s=0}= F'(u, s)\big |_{s=0} - F'(-v, -s)\big |_{s=0} \ .
\]
Since
\[
F'(z, s)\big|_{s=0} = -d \big(f(z, s)\big)^{-D-1}f'(z, s) \big|_{s=0}
 = \frac{D}2  (1-z) \ ,
\]
we have
\[
\g'(u,v;0) = -\frac{D}2(u +v) = -D\,x_0\, .
\]
\end{proof}
\bigskip

\noindent
{\bf Errata.} 
This paper was originally part of a larger manuscript ``The massive modular Hamiltonian". We are indebted to H.  Bostelmann, D.  Cadamuro and K.  Sander for having pointed out a gap concerning the positive mass case there. Therefore, the massive analysis in \cite[Sections 5.2, 5.3, 5.4]{L21} is to be ignored. 

\bigskip

\noindent
{\bf Acknowledgements.} 
An initial part of this work was done in June 2019 during the
program ``Operator Algebras and Quantum Physics'' 
at the Simons Center for Geometry and Physics at Stony Brook.
We thank D. Buchholz for conversations. 

\smallskip\noindent
We acknowledge the MIUR Excellence Department Project awarded to the Department of Mathematics, University of Rome Tor Vergata, CUP E83C18000100006.

\end{document}